\documentclass[onecolumn,10pt]{IEEEtran}
\usepackage{amsfonts,amsmath,amssymb,amsthm,caption,algorithmic}
\usepackage{amsbsy}
\usepackage{graphicx,multirow,bm}
\usepackage{color}
\makeatletter
\newtheoremstyle{mythm}{3pt}{3pt}{}{16pt}{\bfseries}{:}{.5em}{}
\theoremstyle{mythm}
\newtheorem{theorem}{Theorem}
\newtheorem{example}{Example}
\newtheorem{definition}{Definition}
\newtheorem{remark}{Remark}

\newtheorem{lemma}{Lemma}

\begin{document}
\title{Coded Caching Schemes with Low Rate and Subpacketizations
\author{Minquan Cheng, Qifa Yan, Xiaohu Tang, \IEEEmembership{Member,~IEEE}, Jing Jiang
}
\thanks{M. Cheng and J. Jiang are with Guangxi Key Lab of Multi-source Information Mining $\&$ Security, Guangxi Normal University,
Guilin 541004, China (e-mail: $\{$chengqinshi,jjiang2008$\}$@hotmail.com).}
\thanks{Q. Yan and X. Tang are with the Information Security and National Computing Grid Laboratory,
Southwest Jiaotong University, Chengdu, 610031, China (e-mail: qifa@my.swjtu.edu.cn, xhutang@swjtu.edu.cn).}
}
\date{}
\maketitle

\begin{abstract}
Coded caching scheme, which is an effective technique to increase the transmission efficiency during peak traffic times, has recently become quite popular among the coding community. Generally rate can be measured to the transmission in the peak traffic times, i.e., this efficiency increases with the decreasing of rate. In order to implement a coded caching scheme, each file in the library must be split in a certain number of packets. And this number directly reflects the complexity of a coded caching scheme, i.e., the complexity increases with the increasing of the packet number. However there exists a tradeoff between the rate and packet number. So it is meaningful to characterize this tradeoff and design the related Pareto-optimal coded caching schemes with respect to both parameters.

Recently, a new concept called placement delivery array (PDA) was proposed to characterize the coded caching scheme. However as far as we know no one has yet proved that one of the previously known PDAs is Pareto-optimal. In this paper, we first derive two lower bounds on the rate under the framework of PDA. Consequently, the PDA proposed by Maddah-Ali and Niesen is Pareto-optimal,
and a tradeoff between rate and packet number is obtained for some parameters.
Then, from the above observations and the view point of combinatorial design, two new classes of Pareto-optimal PDAs are obtained. Based on these PDAs, the schemes with low rate and packet number are obtained. Finally the performance of some previously known PDAs are estimated by comparing with these two classes of schemes.
\end{abstract}

\begin{IEEEkeywords}
Coded caching scheme, rate, packet number, placement delivery array, lower bound, Pareto-optimal.
\end{IEEEkeywords}
\section{Introduction}
Recently, as the wireless data traffic is increasing at an incredible rate dominated by the video streaming, the wireless network has been imposed a tremendous pressure on the data transmission \cite{White}. Consequently the communication systems are always congested during the peak-traffic times. So reducing this congestion is very meaningful in wireless network and now is a hot topic in industrial and academic fields. Caching system, which proactively caches some contents at the network edge during off-peak times, is a promising solution to reduce congestion % by shifting traffic from peak to off-peak times
 (see \cite{AA,GMDC,GGMG,JTLC,AN}, and references therein).

An important caching system, centralized caching system, is widely studied. In this system, the following scenario are always focused: a single server containing $N$ files with the same length connects to $K$ users over a shared link and each user has a cache memory of size $M$ files (see Fig. \ref{system}).
 \begin{figure}[htbp]
\centering\includegraphics[width=0.4\textwidth]{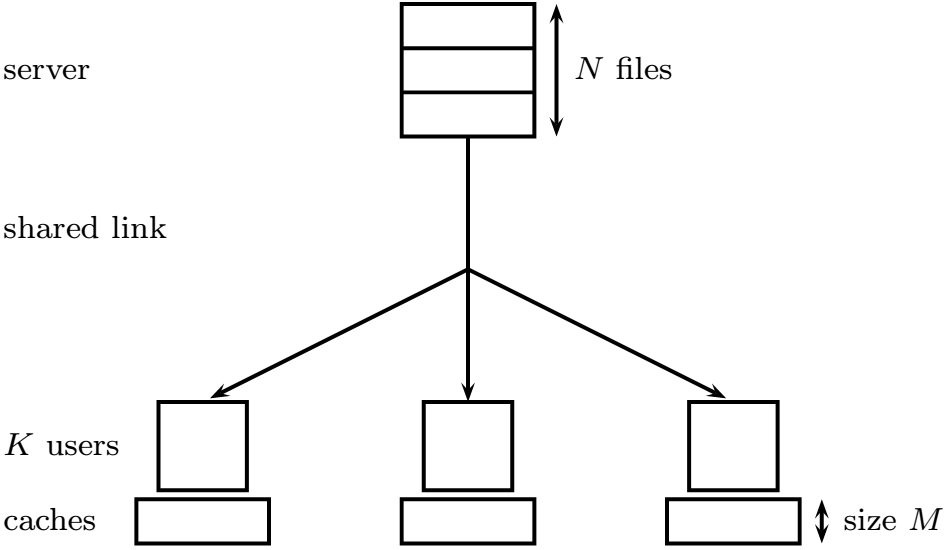}
\caption{Coded caching system}\label{system}
\end{figure}
We call it $(K,M,N)$ caching system in this paper. In the caching system, a caching scheme consists of two phases: a placement phase during off-peak traffic times and a delivery phase during peak traffic times. In the placement phase, the server proactively sends some contents in each user's cache. In the delivery phase, each user requests for one file from server. And after receiving the requests, server transmits some data to users such that each user's demand is satisfied with the help of its caches. Since the placement phase is carried out without any knowledge of the user requests, it is interested in studying the caching scheme such that the broadcasted amount, which can be satisfied for all possible user requests, is as small as possible in the delivery phase. This broadcasted amount is called the rate. Clearly {\em the first objective is to minimize the rate $R$} since it represents the efficiency of a caching scheme.

In the seminal work \cite{AN}, a coded caching scheme based on network coding theory was proposed for the centralized caching system. They showed that the rate of a coded caching scheme can be effectively further reduced by jointly designing the content placement phase and the content delivery phase¡£
%, i.e., various demands of all users can be satisfied using the smaller size of contents broadcasted in the delivery phase and the caches in the placement phase.
They also designed the first determined scheme for a $(K,M,N)$ coded caching system. Such a scheme is referred to as MN scheme in this paper. According to an elaborate uncoded placement and a coded delivery, they showed that the $(K,M,N)$ MN scheme is able to reduce the rate $R$ from $K(1-\frac{M}{N})$ of uncoded caching scheme to $K(1-\frac{M}{N})/(1+\frac{KM}{N})$.
When $N<K$, by employing coded multicasting opportunities across users with the same demand, an improved coded caching scheme is presented in \cite{T}. When $K\leq N$, by means of graph theory, reference \cite{WTP} showed that MN scheme has minimum delivery rate under the constraint of uncoded cache placement. So MN scheme has been also extensively employed in practical scenarios, such as decentralized version \cite{MN1}, device to device networks \cite{JCM}, online
caching update \cite{RAN,YUTC} and hierarchical networks \cite{KNMD}, \cite{Xiao2016tree} and so on. There are also many studies following \cite{AN}, for instances, \cite{AG,GR,STC,TR,WLG,YTC} etc.

%They showed that the rate can be effectively further reduced by jointly designing the content placement phase and the content delivery phase, i.e., various demands of all users can be satisfied using the smaller size of contents broadcasted in the delivery phase and the caches in the placement phase. Since the placement phase is carried out without any knowledge of the user requests, it is interested in optimizing the strategies in the placement phase and delivery phase such that the broadcasted amount is as small as possible in the delivery phase for the worst-case requests. {\color{red}This broadcasted amount in the delivery phase for the worst-case requests is called the rate of the scheme. We denoted it by $R$.} Clearly {\em the objective is to minimize the rate $R$} since it represents the efficiency of a coded caching scheme. A centralized coded caching system is denoted by $(K,M,N)$ coded caching system if the central server containing $N$ files
%connects to $K$ users over a shared link and each user has a cache memory of size $M$ files.

In order to implement a coded caching scheme, Maddah-Ali and Niesen in \cite{AN} showed that each file must be split into a certain number of packets. We denote such a number by $F$. In MN scheme, each file must be divided into ${K\choose KM/N}$ packets where $KM/N$ is an integer. Clearly
%$F$ increases exponentially with the number of users $K$.
this would become infeasible when $K$ is large. Furthermore, the complexity of a coded caching scheme increases with the parameter $F$. So {\em the second objective is to minimize the packets number $F$}, and plays an important role in the field of coded caching scheme. Aiming at this objective, Shanmugam et al. in \cite{SJTLD} first discussed the problem of subpacketization by grouping method.
Recently Yan et al. in \cite{YCTC} characterized the $(K,M,N)$ coded caching scheme with $F$ subpacketization by a very interesting $F\times K$ array which is called $(K,F,Z,S)$ placement delivery array (PDA), where $M/N=Z/F$ and $R=S/F$. Then they proved that MN scheme is equivalent to a special PDA (MN PDA for short throughout this paper). Furthermore, by increasing little delivery rate, they obtained two infinite classes of PDAs such that $F$ reduces significantly comparing with that of MN PDA. By increasing more rate several PDAs were constructed for further reducing the value of $F$ respectively by means of hypergraphs \cite{SZG}, bipartite graphs \cite{YTCC}, Ruzsa-Szem\'{e}redi graphs \cite{STD}, resolvable design \cite{TR} and so on. From the above introductions, we can see that there is a tradeoff between rate and packet number. Clearly it is very interesting to characterize this tradeoff and construct the related Pareto-optimal coded caching schemes with respect to both parameters.
However as far as we know no one has yet proved that one of the previously known PDAs is Pareto-optimal.

In this paper, we focus on Pareto-optimal PDAs. Firstly two lower bounds on the value of $R$ are derived. Consequently, MN PDA is a pareto optimal, and a tradeoff between $R$ and $F$ is obtained for some parameters. Secondly, unlike the previously know strategies of constructing PDAs, we use a different strategy, i.e., we characterize PDAs by means of a set of $3$ dimensional vectors. From this characterization and the above lower bounds, two new classes of Pareto-optimal PDAs are obtained. Based on these PDAs, the schemes with low rate and packet number are obtained when $K\leq N$. Finally the performance of three previously known schemes are estimated by comparing with these two classes of schemes.

The rest of this paper is organized as follows. Section \ref{preliminaries} briefly reviews the relationship between the coded caching scheme proposed in \cite{AN} and concept of the PDA introduced in \cite{YCTC}. In Section \ref{Lower bounds}, two lower bounds on the value of $R$ are derived. In Section \ref{combinatorial designs}, we characterize a PDA by means of a set of $3$ dimensional vectors and obtain two new classes of PDAs by the MN PDA. In Section \ref{se-other two good PDAs} we show that these two class of PDAs are Pareto-optimal. Finally for some parameters $K$, $M/N$ and $R$, three comparisons with previously known PDAs are proposed. Conclusion is drawn in Section \ref{conclusion}.

\section{Preliminaries}
\label{preliminaries}

In this paper, we denote arrays by bold capital letters, and assume that each entry has exactly one symbol. We use $[a,b]=\{a,a+1,\ldots,b\}$ and $[a, b)=\{a,a+1,\ldots,b-1\}$ for any integers $a$ and $b$ with $a\leq b$.
%In this paper, we will use the notations and the assumptions proposed by Maddah-Ali and Niesen in \cite{AN}.
\subsection{Coded caching schemes}
In a $(K,M,N)$ caching system, we denote the $N$ files by $\mathcal{W}=\{W_0,W_1,\ldots,W_{N-1}\}$ and $K$ users by $\mathcal{K}=\{0,1,\ldots,K-1\}$. A coded caching scheme consists of two separated two phases which are introduce by  Maddah-Ali and Niesen in \cite{AN} as follows:
\begin{itemize}
\item {\bf Placement phase}: Firstly each file is subdivided into $F$ equal packets, \emph{i.e.}, $W_{i}=\{W_{i,j}:j\in[0,F)\}$. Such a coded caching scheme is called $F$-division scheme; Then Each user just caches some part of the packets of all files (This is also called uncoded placement phase), and uses the same caching policy for all files, i.e., each user caches packets with the same indices from all files, where packets belonging to every file is ordered according to a chosen numbering. Denote $\mathcal{Z}_k$ the packets cached by user $k$.
\item {\bf Delivery phase}: Each user randomly requests one file from the files set $\mathcal{W}$ independently. The request is denoted by $\mathbf{d}=(d_0,d_1,\cdots,d_{K-1})$, which indicates that user $k$ requests the $d_k$-th file $W_{d_k}$ for any $d_k\in[0,N)$ and $k\in\mathcal{K}$. Once the server received the users' request $\mathbf{d}$, it broadcasts a signal of at most $RF$ packets to users. Each user is able to recover its requested file from the received signal with the help of the contents in its own cache.
\end{itemize}
By the way the assumptions in placement phase are very useful and have been used in many references. Since the placement phase is carried out without any knowledge of the user requests, we prefer to construct an $F$-division $(K,M,N)$ caching scheme such that the rate $R$ is as small as possible for any request ${\bf d}\in [0,N)^{K}$. By means of graph theory, reference \cite{WTP} showed that the rate of MN scheme is the minimum rate in the uncoded cache placement. However the packet number $F={K\choose KM/N}$ is too large when $K$ is large. In order to study the packet number of a coded caching scheme, an interesting combinatorial structure called placement delivery array, which is introduced in the following subsection, was proposed in \cite{YCTC}.
%In addition, the complexity of a coded caching scheme increases with the packet number $F$. So we prefer the packet number of a coded caching scheme as small as possible. This is our second objective in coded caching scheme which becomes an important role in the field of coded caching scheme recently. Overall it is meaningful to construct a scheme with rate and packet number as small as possible.

 %In order to simplify the design, Maddah-Ali and Niesen considered the uncoded placement phase and assumed that in \cite{AN} assume that identical caching policy\footnote{This assumption is used in most of studies on coded caching schemes such as
%    \cite{AG,JCM,KNMD,AN,MN1,SZG,SJTLD,STD,YCTC} and on.} for all files is assumed for each user, $i.e.$, each user caches packets with the same indices from all files under the uncoded placement phase, where packets belonging to every file is ordered according to a chosen numbering.
%And the assumptions in the above two phases are referred to as MN assumptions for introduction easily.

\subsection{Placement delivery array}
\begin{definition}\rm(\cite{YCTC})
\label{def-PDA}
For positive integers $K$ and $F$, an $F\times K$ array $\mathbf{P}=(p_{i,j})$, $i\in [0,F), j\in[0,K)$, composed of a specific symbol $``*"$ called star and $S$ nonnegative integers
$0,1,\cdots, S-1$, is called a $(K,F,S)$ placement delivery array (PDA) if it satisfies C$1$ in the following conditions:
\begin{enumerate}
  \item [C$1$.] For any two distinct entries $p_{i_1,j_1}$ and $p_{i_2,j_2}$,    $p_{i_1,j_1}=p_{i_2,j_2}=s$ is an integer  only if
  \begin{enumerate}
     \item [a.] $i_1\ne i_2$, $j_1\ne j_2$, i.e., they lie in distinct rows and distinct columns; and
     \item [b.] $p_{i_1,j_2}=p_{i_2,j_1}=*$, i.e., the corresponding $2\times 2$  subarray formed by rows $i_1,i_2$ and columns $j_1,j_2$ must be of the following form
  \begin{eqnarray*}
    \left(\begin{array}{cc}
      s & *\\
      * & s
    \end{array}\right)~\textrm{or}~
    \left(\begin{array}{cc}
      * & s\\
      s & *
    \end{array}\right).
  \end{eqnarray*}
   \end{enumerate}
  \end{enumerate}
\end{definition}
For any positive integer $Z\leq F$, $\mathbf{P}$ is denoted by $(K,F,Z,S)$ PDA if
\begin{enumerate}
\item [C$2$.] each column has exactly $Z$ stars.
   \end{enumerate}
And for any positive integer $g$, a $(K,F,Z,S)$ PDA is said to be $g$-regular, denoted by $g$-$(K,F,Z,S)$ PDA, if
 each integer of $[0,S)$ appears exactly $g$ times.
\begin{example}\rm
\label{E-pda}
It is easy to verify that the following array is a $3$-$(4,6,3,4)$ PDA:
\begin{eqnarray*}
\mathbf{P}_{6\times 4}=\left(\begin{array}{cccc}
*&*&0&1\\
*&0&*&2\\
*&1&2&*\\
0&*&*&3\\
1&*&3&*\\
2&3&*&*
\end{array}\right).
\end{eqnarray*}
\end{example}

Given a $(K,F,Z,S)$ PDA $\mathbf{P}=(p_{j,k})_{j\in [0,F),k\in[0,K)}$, Yan \textit{et al.} in \cite{YCTC} showed that an $F$-division $(K,M,N)$ coded caching scheme with $M/N=Z/F$ and $R=S/F$ can be realized by the following rule:
\begin{itemize}
\item [1.] \textbf{Placement Phase:} %All the files are cached in the same manner.
Each file is divided into $F$ equal packets, i.e., $W_i=(W_{i,0}, W_{i,1}, \ldots, W_{i,F-1})$, $i=0$, $1$, $\ldots$, $N-1$. Then user $ k\in\mathcal{K}$ caches the packets
      \begin{align}
      \mathcal{Z}_k=\{W_{i,j}: p_{j,k}=*,\ i\in[0,N)\}.\label{eq_p1}
      \end{align}
     Clearly each user stores $N\cdot Z$ packets by condition C$2$. So the whole size of cache is $N\frac{Z}{F}=N\frac{M}{N}=M$, which satisfies the users' cache constraint.
\item [2.] \textbf{Delivery Phase:}  Once the server receives the request $\mathbf{d}=(d_0,d_1,\cdots,d_{K-1})$, at the  time slot $s$, $0\le s<S$, it sends
      \begin{align}
      \bigoplus_{p_{j,k}=s,  j\in[0, F),  k\in[0,K)}W_{d_{k},j}\label{eq_p2}
      \end{align}
\end{itemize}
%Assume that in $\mathbf{P}$ there are $g$ entries $p_{j_1,k_1}=p_{j_2,k_2}=\cdots=p_{j_g,k_g}=s$ where $0\le j_1,\cdots,j_g< F$ and
%$0\le k_1,\cdots,k_g<K$.  Consider the sub-array  formed by rows $j_1,\cdots,j_g$ and columns $k_1,\cdots,k_g$, which is of order $g\times g$ since $j_{h}\ne j_{l}$ and $k_h\ne k_l$ for all $1\le h\ne l\le g$ by C$1$-a. Further, applying C$1$-b we have $p_{j_h,k_l}=*$ for all $1\le h\ne l\le g$. That is to say,
%this sub-array is equivalent to the following $g\times g$ array
%\begin{eqnarray}\label{Eqn_Matrix_1}
%    \left(\begin{array}{cccc}
%      s & * & \cdots & *\\
%      * & s & \cdots & *\\
%      \vdots & \vdots &\ddots & \vdots\\
%      * & * & \cdots & s
%    \end{array}\right)
%\end{eqnarray}
%with respect to row/column permutation. For instance, in Example \ref{E-pda}, there are $3$ entries $p_{3,0}=p_{1,1}=p_{0,2}=0$ of $\mathbf{P}_{6\times 4}$. Then the following subarray formed by rows $0,1,3$ and columns $2,1,0$ can be obtained.
%\begin{eqnarray*}
%\mathbf{P}^{(0)}=\left(\begin{array}{ccc}
%0      & *       &*\\
%*      & 0       &*\\
%*      &*        &0
%\end{array}\right)
%\end{eqnarray*}
%According to \eqref{eq_p2}, at the  time slot $s$, $0\le s<S$, the sever sends
%$\bigoplus_{1\le h\le g}W_{d_{k_h},j_h}$.
%Note from \eqref{Eqn_Matrix_1} that in column $l$, all the entries are $``*"$s except for the $l$-th one. Then
%it follows from \eqref{eq_p1} that user $k_l$ has already all the other packets $W_{d_{k_h},j_h}$, $1\le h\ne l\le g$, in its cache at the placement phase. Then,
%it can easily decode the desired packet $W_{d_{k_l},j_l}$.
The rate of the scheme is $R=S/F$ for any request ${\bf d}$ since each file is divided into $F$ equal packets and there are $S$ distinct integers in $\mathbf{P}$.

\begin{example}\label{exam2} Given a $(4,6,3,4)$ PDA in Example \ref{E-pda}, from the above analysis one can obtain a $6$-division $(4,3,6)$ coded caching scheme in the following way.
\begin{itemize}
   \item \textbf{Placement Phase}: First the files are denoted by $W_0,W_1,W_2,W_3,W_4,W_5$ respectively, and each file is divided into $F=6$ packets, $i.e.$ $W_i=\{W_{i,0},W_{i,1},W_{i,2},W_{i,3},W_{i,4},W_{i,5}\}$ , $i\in [0,6)$. Then the contents in each user are
       \begin{align*}
       \mathcal{Z}_0=\left\{W_{i,0},W_{i,1},W_{i,2}:i\in[0,6)\right\}\ \ \ \ \ \ \
       \mathcal{Z}_1=\left\{W_{i,0},W_{i,3},W_{i,4}:i\in[0,6)\right\} \\
       \mathcal{Z}_2=\left\{W_{i,1},W_{i,3},W_{i,5}:i\in[0,6)\right\}\ \ \ \ \ \ \
       \mathcal{Z}_3=\left\{W_{i,2},W_{i,4},W_{i,5}:i\in[0,6)\right\}
       \end{align*}
   \item \textbf{Delivery Phase}: Assume that the request vector is $\mathbf{d}=(0,1,2,3,4,5)$. Table \ref{table1} shows the transmitting process.
   \begin{table}[!htp]

  \normalsize{
  \begin{tabular}{|c|c|}
\hline
    % after \\: \hline or \cline{col1-col2} \cline{col3-col4} ...
   Time Slot& Transmitted Signnal  \\
\hline
   $0$&$W_{0,3}\oplus W_{1,1}\oplus W_{2,0}$\\ \hline
   $1$&$W_{0,4}\oplus W_{1,2}\oplus W_{3,0}$\\ \hline
  $2$& $W_{0,5}\oplus W_{2,2}\oplus W_{3,1}$\\ \hline
   $3$& $W_{1,5}\oplus W_{2,4}\oplus W_{3,3}$\\ \hline
% \bottomrule
  \end{tabular}}\centering
  \caption{Delivery steps in Example \ref{exam2} }\label{table1}
\end{table}
\end{itemize}
\end{example}
%By the above analysis and example, the following statement is proved in \cite{YCTC}.
\begin{theorem}(\cite{YCTC})\label{thm1}
%For any positive integers $K$, $F$, $M$, $N$ and $R$, an $F$-division $(K,M,N)$ coded caching scheme with $R$ under MN assumptions exists if and only if there exists a $(K,F,Z,S)$ PDA with $M/N=Z/F$ and $R=S/F$.
Given a $(K,F,Z,S)$ PDA $\mathbf{P}=(p_{j,k})_{F\times K}$, one can obtain a corresponding $F$-division caching scheme for any $(K,M,N)$ caching system with $M/N=Z/F$. Precisely, each user is able to decode its requested file correctly for any request $\mathbf{d}$ at the rate
$R=S/F$.
\end{theorem}
\begin{theorem}({\em MN PDA},\cite{YCTC})\label{th-AN-Y}
MN scheme is equivalent to a $(t+1)$-$(K,{K\choose t},{K-1\choose t-1},{K\choose t+1})$ PDA with $t$ stars in each row where $t=KM/N$.
\end{theorem}
When $K$ is very large, Shanmugam et al. in \cite{SJTLD} proposed a grouping method to reduce $F$. That is, for some integers $K'$ and $K$, assume that $K'|K$ (for the sake of simplicity). First we divide $K$ users into $\frac{K}{K'}$ groups with the same size, then we use MN PDA for each group. So the following result can be obtained.
\begin{lemma}
\rm(\cite{SJTLD})\label{lem-shanmugam}
For any positive integers $k$, $m$ and $t$ with $0< t< k$, there exists an $(mk,{k \choose t},{k-1 \choose t-1},m{k\choose t+1})$ PDA generated by MN PDA in Theorem \ref{th-AN-Y}.
\end{lemma}
In order to reduce the packet number efficiently while the rate increases a little, the following PDA was constructed.
\begin{lemma}\rm(\cite{YCTC})\label{lem-yan}
For any positive integers $m$ and $q\geq2$, there exists a $(q(m+1),(q-1)q^m,(q-1)^2q^{m-1},q^m)$ PDA.
\end{lemma}
By increasing more $R$, there are some other PDAs with smaller $F$ constructed in a different light, such as resolvable design \cite{TR}, hypergraphs \cite{SZG}, strongly bipartite graphs \cite{YTCC}, Ruzsa-Szem\'{e}redi graphs \cite{STD} and so on. For example, the following result is obtained from the view point of hypergraphs.
\begin{lemma}\rm(\cite{SZG})\label{lem-shang}
There exists %an ${a+b\choose a}$-$({n\choose b}$,${n\choose a}$, ${n,\choose a}-{n-b\choose a},{n\choose a+b})$PDA with ${n\choose b}-{n-a\choose b}$ for any positive integers $a$, $b$ and $n$ with $a+b\le n$, and
 an $({m\choose l}q^l,q^m(q-1)^l,(q^m-q^{m-l})(q-1)^l,q^m)$ PDA for any positive integers $q\geq 2$, $l$ and $m$ with $l\le m$.
\end{lemma}
From the above illustrations, we can study an $F$-division $(K,M,N)$ coded caching scheme by means of a $(K,F,Z,S)$ PDA where $M/N=Z/F$ and $R=S/F$. And we can see that there is a tradeoff between rate and packet number intuitively. For convenience, a PDA is called Pareto-optimal if the scheme generated by it achieves the tradeoff between $R$ and $F$ among the schemes which can realized by PDAs. So in this paper we will consider the tradeoff between the lower bound on $R$ and $F$, and then obtain some related classes of Pareto-optimal PDAs. Based on these Pareto-optimal PDAs, we can obtain the corresponding schemes with low $R$ and $F$ when $K\leq N$.

\section{Lower bound on $S$}
\label{Lower bounds}
In this section, we derive two lower bounds on $S$ for given positive integers $K$, $F$ and $Z$ of a PDA. According to these lower bounds, we have that MN PDA is Pareto-optimal, and there exactly exists a tradeoff between the lower bound on $R$ and $F$.

\subsection{The first lower bound on $S$ }
Now let us consider our first lower bound on $S$ in the following. Let $\mathbf{P}$ be a $(K,F,S)$ PDA. For any integer $s\in [0,S)$, assume that there are $r_s$ entries in all, say $p_{j_u,k_u}$, $1\leq u\leq r_s$, $0\le j_u< F$ and
$0\le k_u<K$, such that $p_{j_u,k_u}=s$. Consider the subarray formed by rows $j_1,\cdots,j_{r_s}$ and columns $k_1,\cdots,k_{r_s}$, which is of order  $r_s\times r_s$ since $j_{u}\ne j_{v}$ and $k_u\ne k_v$ for all $1\le u\ne v\le r_s$ from the definition of a PDA. Further, we have $p_{j_u,k_v}=*$ for all $1\le u\ne v\le r_s$. That is to say,
this subarray is equivalent to the following $r_s\times r_s$ array
\begin{eqnarray}\label{Eqn_Matrix_1_1}
\mathbf{P}^{(s)}=\left(\begin{array}{cccc}
      s & * & \cdots & *\\
      * & s & \cdots & *\\
      \vdots & \vdots &\ddots & \vdots\\
      * & * & \cdots & s
    \end{array}\right)
\end{eqnarray}
with respect to row/column permutation.

For instance, in Example \ref{E-pda}, there are $3$ entries $p_{3,0}=p_{1,1}=p_{0,2}=0$ of $\mathbf{P}_{6\times 4}$. Then the following subarray formed by rows $0,1,3$ and columns $2,1,0$ can be obtained.
\begin{eqnarray*}
\mathbf{P}^{(0)}=\left(\begin{array}{ccc}
0      & *       &*\\
*      & 0       &*\\
*      &*        &0
\end{array}\right)
\end{eqnarray*}

\begin{theorem}
\label{th-lower bound}
If there exists a $(K,F,S)$ PDA $\mathbf{P}$ with $n>0$ integer entries and integer set $[0,S)$, then
\begin{eqnarray}
\label{eq-Lowebound-K<F}
\frac{nF}{KF+F-n}\leq S.
\end{eqnarray}
Further the above equality holds if and only if each row has $n/F$ integer entries and each integer $s\in [0,S)$ occurs exactly $n/S$ times.
\end{theorem}
\begin{proof}
Assume that an integer $s\in [0,S)$ occurs exactly $r_s$ times in $\mathbf{P}$. From \eqref{Eqn_Matrix_1_1}, there are $r_s(r_s-1)$ stars in $\mathbf{P}^{(s)}$. Then the total number of stars in $\mathbf{P}^{(s)}$, $s=0$, $1$, $\ldots$, $S-1$, is
\begin{eqnarray}
\label{eq-cha4-1}
W=\sum\limits_{s=0}^{S-1}r_s(r_s-1).
\end{eqnarray}

On the other hand, we can estimate the value of $W$ in a different way.
%Since the time of every integer in $[0,S)$ occurring is less than the minimum of $K$ and $F$,
Without loss of generality, assume that $K<F$, and row $i$ has exactly $r'_i$ integer entries, denoted by $s_{1,i}$, $\ldots$, $s_{r'_i,i}$. Then each star of the $i$-th row occurs in at most $r'_i$ times in $\mathbf{P}^{(s_{1,i})}$, $\ldots$, $\mathbf{P}^{(s_{r'_i,i})}$ since it occurs in $\mathbf{P}^{(s_{h,i})}$ at most once for any $h\in [1,r'_i]$. So the total number of occurrences of all the stars of $\mathbf{P}$ in $\mathbf{P}^{(s)}$, $s=0$, $1$, $\ldots$, $S-1$, is at most
\begin{eqnarray}
\label{eq-cha4-2}
W'=\sum\limits_{i=0}^{F-1}r'_i(K-r'_i).
\end{eqnarray}
Clearly $W\leq W'$, i.e.,
\begin{eqnarray*}
\sum\limits_{s=0}^{S-1}r_s(r_s-1)\leq\sum\limits_{i=0}^{F-1}r'_i(K-r'_i).
\end{eqnarray*}
This implies,
\begin{eqnarray}
\label{eq-IN 1}
\sum\limits_{s=0}^{S-1}r^2_s+\sum\limits_{i=0}^{F-1}{r'}^2_i\leq Kn +n.
\end{eqnarray}
Moreover,
\begin{eqnarray}
\label{eq-IN 2}
\sum\limits_{s=0}^{S-1}r^2_s\geq\frac{1}{S}\left(\sum\limits_{s=0}^{S-1}r_s\right)^2=\frac{n^2}{S}, \ \ \ \
\sum\limits_{i=0}^{F-1}{r'}^2_i\geq\frac{1}{F}\left(\sum\limits_{i=0}^{F-1}r'_i\right)^2=\frac{n^2}{F},
\end{eqnarray}
where the equalities hold if and only if $r_0=r_1=\ldots=r_{S-1}$ and $r'_0=r'_1=\ldots=r'_{F-1}$ respectively. Combining \eqref{eq-IN 1} and \eqref{eq-IN 2}, the inequality \eqref{eq-Lowebound-K<F} can be obtained. And the equality in \eqref{eq-Lowebound-K<F} holds if and only if $r_0=r_1=\ldots=r_{S-1}=\frac{n}{S}$ and $r'_0=r'_1=\ldots=r'_{F-1}=\frac{n}{F}$ are positive integers.
\end{proof}

From Theorem \ref{th-lower bound}, the following statement can be obtained.
\begin{lemma}
Given a $(K,F,Z,S)$ PDA $\mathbf{P}$, $S=\frac{(F-Z)KF}{KZ+F}$ if and only if it is a $g$-regular where $g=KZ/F+1$.
\end{lemma}
\begin{proof}
Clearly $n=(F-Z)K$. When $S=\frac{(F-Z)KF}{KZ+F}$, i.e., $S=\frac{nF}{KF+F-n}$, from Theorem \ref{th-lower bound} the number of occurrence of each integer in $[0,S)$ is
\begin{eqnarray*}
g=\frac{n}{S}=\frac{(F-Z)K}{\frac{nF}{KF+F-n}}=\frac{KF+F-n }{F}=\frac{KZ}{F}+1.
\end{eqnarray*}

Conversely if $\mathbf{P}$ is $(\frac{KZ}{F}+1)$ regular, we have $S(\frac{KZ}{F}+1)=n$ by counting the number of integer entries. So
$$S=\frac{n}{1+\frac{KZ}{F}}=\frac{nF}{F+KZ}=\frac{K(F-Z)F}{KF+F-K(F-Z)}=\frac{(F-Z)KF}{Z}.$$
\end{proof}
\begin{lemma}(\cite{YCTC})
\label{le-lower-F}For a $g$-$(K,F,Z,S)$ PDA, if $g=KZ/F+1$, then
$F\geq{K\choose KZ/F}$.
\end{lemma}
From Theorem \ref{th-lower bound} and Lemma \ref{le-lower-F}, the following result can be obtained.
%\begin{lemma}
%\label{le-optimal PDA}
%For a $(K,F,Z,S)$ PDA with $S=\frac{(F-Z)KF}{Z}$, then
%$F\geq{K\choose KZ/F}$.
%\end{lemma}
%So from Lemma \ref{le-optimal PDA}, the following result holds.
\begin{theorem}
\label{th-AN-optimal PDA}
For a $(K,F,Z,S)$ PDA with $S=\frac{(F-Z)KF}{KZ+F}$, then $F\geq{K\choose KZ/F}$.
\end{theorem}

When $S=\frac{(F-Z)KF}{KZ+F}$, i.e., $R=S/F=\frac{(F-Z)K}{KZ+F}=K(1-\frac{M}{N})/(K\frac{M}{N}+1)$. This is exactly the rate of MN PDA.
\begin{remark}
\label{re2}
In \cite{YCTC}, the authors showed that an MN PDA has the minimum number of rows among the $g$-regular PDAs with the same parameters $K$, $M/N$ and $R=K(1-\frac{M}{N})/(\frac{M}{N}+1)$. From Theorem \ref{th-AN-optimal PDA}, we have that MN PDA is Pareto-optimal.
\end{remark}

\subsection{The second lower bound on $S$}
\begin{theorem}\label{th-digui-bound}
Given any positive integers $K,F,Z$ with $0<Z<F$, if there exists a $(K,F,Z,S)$ PDA, then
\begin{eqnarray}\label{eq-digui-bound}
S\geq \left\lceil\frac{(F-Z)K}{F}\right\rceil+\left\lceil\frac{F-Z-1}{F-1}\left\lceil\frac{(F-Z)K}{F}\right\rceil\right\rceil+\ldots+
\left\lceil\frac{1}{Z+1}\left\lceil\frac{2}{Z+2}\left\lceil\ldots\left\lceil\frac{(F-Z)K}{F}\right\rceil \ldots\right\rceil\right\rceil\right\rceil.
\end{eqnarray}
\end{theorem}

\begin{proof}
%Since $\mathcal{S}(K,F,F)=0$ is clear, we only need to prove \eqref{eq-digui-bound} for $F>Z$.
Suppose that $\mathbf{P}$ is a $(K,F,Z,S)$ PDA.
Totally, there are $(F-Z)K$ integers in this array. Thus, among the $F$ rows,
there must exist one row containing at lest $\left\lceil \frac{(F-Z)K}{F} \right\rceil$ integers.
Without loss of generality, assume that these $\left\lceil \frac{(F-Z)K}{F} \right\rceil$ integers, say $0$, $1$, $\ldots$, $\left\lceil \frac{(F-Z)K}{F} \right\rceil-1$, are in the first row, i.e.,
\begin{eqnarray*}
\label{eq-digui}
\mathbf{P}=\left(\begin{array}{ccccc|c}
0      & 1     & 2     & \cdots & \left\lceil \frac{(F-Z)K}{F} \right\rceil-1 &\  \cdots\ \ \\ \hline
              &       &  &   &   &\\
       &       &&  \mathbf{P}' &   &\ \ddots\ \  \\
              &       &  &   &   &
                   \end{array}
  \right).
\end{eqnarray*}
If not, we can get such form  of $\mathbf{P}$ by row/column permutations.
Clearly $\mathbf{P}'$ is an $(\left\lceil \frac{(F-Z)K}{F} \right\rceil, F-1, Z, S')$ PDA for one nonnegative integer $S'$. Further, we
know that all the integers in  $[0, \left\lceil \frac{(F-Z)K}{F} \right\rceil)$
do not appear in $\mathbf{P}'$.  Otherwise, it would contradict Property C2. Therefore,  we have
\begin{eqnarray}\label{eq-diyiceng-bound}
S\ge \left\lceil \frac{(F-Z)K}{F} \right\rceil+S'.
\end{eqnarray}

Perform  the same argument to $P'$ till that the remaining array is a $(\left\lceil\frac{1}{Z+1}\left\lceil\frac{2}{Z+2}\left\lceil\ldots\left\lceil\frac{(F-Z)K}{F}\right\rceil \ldots\right\rceil\right\rceil\right\rceil,Z,Z,0)$ PDA. Then, the bound in \eqref{eq-digui-bound} follws by recursively applying
\eqref{eq-diyiceng-bound} $F-Z$ times.
\end{proof}

\begin{example}\label{ex-(6,8,5,5)PDA}
When $K=6$, $F=8$ and $Z=5$, $\left\lceil\frac{(F-Z-1)}{F-1}\left\lceil\frac{(F-Z)K}{F}\right\rceil\right\rceil=1$. So we have $S\geq 3+3-1=5$. It is easy to check that the following array is an optimal $(6,8,5,5)$ PDA .
\begin{eqnarray*}
\mathbf{P}_{8\times 6}=\left(\begin{array}{cccccc}
 0    &*     &*    &*   &3   &*\\
 1    &3     &*    &*   &*   &4\\
 *    &0     &1    &*   &*   &*\\
 2    &*     &3    &*   &*   &*\\
 *    &2     &*    &1   &*   &*\\
 *    &*     &4    &0   &2   &*\\
 *    &*     &*    &*   &1   &0\\
 *    &*     &*    &3   &*   &2
\end{array}\right)
\end{eqnarray*}
\end{example}

From \eqref{eq-digui-bound}, we have
\begin{eqnarray}\label{eq-digui-bound-Lower}
\begin{split}
S&\geq \frac{(F-Z)K}{F}+\frac{F-Z-1}{F-1}\frac{(F-Z)K}{F}+\ldots+
\frac{1}{Z+1}\frac{2}{Z+2}\ldots\frac{(F-Z)K}{F}\\
&=K\frac{F-Z}{F}\left(1+\frac{F-Z-1}{F-1}+\cdots+
\frac{1}{Z+1}\frac{2}{Z+2}\cdots\frac{F-Z-1}{F-1}\right)\\
&\geq \frac{(F-Z)K}{Z+1}
\end{split}
\end{eqnarray}
since \begin{eqnarray*}
&&1+\frac{F-Z-1}{F-1}+\cdots+\frac{2}{Z+2}\frac{3}{Z+3}\cdots\frac{F-Z-1}{F-1}+
\frac{1}{Z+1}\frac{2}{Z+2}\cdots\frac{F-Z-1}{F-1}\\
&=& 1+\frac{F-Z-1}{F-1}+\cdots+\left(1+\frac{1}{Z+1}\right)\frac{2}{Z+2}\frac{3}{Z+3}\cdots\frac{F-Z-1}{F-1}\\
&=& 1+\frac{F-Z-1}{F-1}+\ldots+\frac{Z+2}{Z+1}\frac{2}{Z+2}\ldots\frac{F-Z-1}{F-1}\\
&=& 1+\frac{F-Z-1}{F-1}+\cdots+\frac{2}{Z+1}\frac{3}{Z+3}\ldots\frac{F-Z-1}{F-1}\\
&=& 1+\frac{F-Z-1}{F-1}+ \frac{F-2}{Z+1}\frac{F-Z-2}{F-2}\frac{3}{Z+3}\frac{F-Z-1}{F-1}\\
&=& 1+\frac{F-Z-1}{F-1}+ \frac{F-Z-2}{Z+1}\frac{F-Z-1}{F-1}\\
&=& 1+\left(1+\frac{F-Z-2}{Z+1}\right)\frac{F-Z-1}{F-1}\\
&=& 1+\frac{F-Z-1}{Z+1}=\frac{F}{Z+1}.
\end{eqnarray*}
Given a positive integer $K$ and a positive real number $M/N$, let $Z/F=M/N$. Then the following inequality can be obtained by \eqref{eq-digui-bound-Lower}.
\begin{eqnarray}
\label{eq-tradeoff}
R=\frac{S}{F}\geq\frac{K(1-\frac{M}{N})}{F\frac{M}{N}+1}
\end{eqnarray}
From the above formula, we can see that the lower bound on $R$ increases with $F$ decreasing for any fixed $K$ and $M/N$.
%From the above discussions, there exactly exists a tradeoff between packet number $F$ and the rate of the coded caching scheme for any fixed $K$ and $M/N$, i.e.,.
And we will show that the equality in \eqref{eq-tradeoff} holds for some fixed parameters in Subsection \ref{the second optimal PDA}.
%As far as we know, this statement is first showed in the field of coded caching scheme.
%However it is noted that this lower is only fit for the schemes which can realized by PDAs.
However we claim that this lower bound is not tight when $F>K$ since reference \cite{WTP} showed that $R=\frac{K(1-\frac{M}{N})}{\frac{KM}{N}+1}$ is the minimum rate. When $F>K$,
$$\frac{K(1-\frac{M}{N})}{\frac{FM}{N}+1}<\frac{K(1-\frac{M}{N})}{\frac{KM}{N}+1}$$
holds for any positive integers $K$, $M$ and $N$. Clearly this is impossible.

Unfortunately according to these two lower bound on $S$ in this section, it is not easy to check whether one of the other know PDAs is Pareto-optimal. But according to the two new classes of PDAs obtained in Section \ref{combinatorial designs}, we can estimate the performance of some known PDAs indirectly in Section \ref{sec-comparison}.

\section{Characterization of PDA from combinatorial designs}
\label{combinatorial designs}
In this section, we characterize a PDA by means of a set of $3$ dimensional vectors. Consequently several classes of PDAs are obtained. For ease of exposition, an $F\times K$ array $\mathbf{P}=(p_{i,j})$, $0\leq i<F$, $0\leq j<K$, on $[0,S)\cup\{*\}$ will be represented by a set of ordered triples %\begin{eqnarray}
%\label{PDA-odered set}
$\mathcal{C}=\{(i, j, p_{i,j})^T\ | p_{i,j}\in [0,S)\}$ in this section.
%\end{eqnarray}
 Clearly an array corresponds to an unique vector set, and the converse is also correct. So we assume that $\mathcal{C}\subseteq [0,F)\times [0,K)\times [0,S)$ and is called the incidence set of $\mathbf{P}$ in the following.
\begin{example}\rm
\label{associated set}
The $(4,6,3,4)$ PDA $\mathbf{P}_{6\times 4}$ in Example \ref{E-pda} can be represented by the following vector set where each column is a vector.
\begin{eqnarray*}
\mathcal{C}=\left(\begin{array}{ccccccccccccccccccccccccccc}
0 & 0 & 1 & 1 & 2 & 2 & 3 & 3 & 4 & 4 & 5 & 5\\
2 & 3 & 1 & 3 & 1 & 2 & 0 & 3 & 0 & 2 & 0 & 1 \\
0 & 1 & 0 & 2 & 1 & 2 & 0 & 3 & 1 & 3 & 2 & 3
\end{array}\right)
\end{eqnarray*}
\end{example}

Let ${\bm x}$ and ${\bm y}$ be elements of $\mathcal{C}$. The (Hamming) distance between ${\bm x}$ and ${\bm y}$, denoted by $d({\bm x},{\bm y})$, is defined to be the number of coordinates at which ${\bm x}$ and ${\bm y}$ differ. And the (minimum) distance of $\mathcal{C}$, denoted by $d(\mathcal{C})$, is
$$d(\mathcal{C})=\hbox{min}\{d({\bm x},{\bm y})\ |\ {\bm x}, {\bm y}\in \mathcal{C}, {\bm x}\neq {\bm y}\}.$$

For instance, $d(\mathcal{C})=2$ in Example \ref{associated set}. Furthermore it is easy to check that the following $\triangle$ is not a subset of $\mathcal{C}$ in Example \ref{associated set}, where $i_1\neq i_2\in [0,6)$, $j_1\neq j_2\in [0,4)$ and $a\neq b\in [0,4)$.
\begin{eqnarray}
\label{fo-r-pda}
\triangle=\left(
                     \begin{array}{ccc}
                       i_1 & i_1 & i_2\\
                       j_1 & j_2 & j_2\\
                       a   & b   & a
                     \end{array}
                   \right)
\end{eqnarray}
In fact, we can show that $\triangle\not\subseteq\mathcal{C}$ from C1. We call such $\triangle$ a forbidden configuration of $\mathcal{C}$.
\begin{theorem}\rm
\label{the-pda-pda}
There exists a $(K,F,S)$ PDA with $n$ integer entries if and only if there exists a set $\mathcal{C}$ with cardinality $n$ satisfying
\begin{itemize}
\item[P1:] the minimum Hamming distance is at least $2$, and
\item[P2:] $\triangle$ in \eqref{fo-r-pda} is the forbidden configuration.
\end{itemize}
\end{theorem}
\begin{proof} Given a $(K,F,S)$ PDA $\mathbf{P}$, denote its incidence set by $\mathcal{C}$. Clearly the number of integer entries equals the order of $\mathcal{C}$. Firstly, assume that there exist two distinct vectors, say $(i_1,j_1,a_1)^T$ and $(i_2,j_2,a_2)^T\in \mathcal{C}$, with distance less than $2$. It is easy to check that at most one of equalities $i_1=i_2$ and $j_1=j_2$ holds since every entry has at most one symbol in $\mathbf{P}$. Without loss of generality we assume that $i_1=i_2$. Then we have $a_1=a_2$. This contradicts C1-a) of Definition \ref{def-PDA}. So P1 holds. Secondly, suppose that $\triangle$ in \eqref{fo-r-pda} is a subset of $\mathcal{C}$. Then we have $p_{i_1,j_1}=p_{i_2,j_2}=a$, $p_{i_1,j_2}=b$, $a,b\in [0,S)$, a contradiction to C1-b) of Definition \ref{def-PDA}. So P2 holds.

Conversely, assume that a set $\mathcal{C}$ with cardinality $n$ satisfies P1 and P2. From P1, for any integers $i\in [0,F)$ and $j\in [0,K)$, there is at most one integer $a\in [0,S)$ such that $(i,j,a)^T\in \mathcal{C}$. Then we can define an $F\times K$ array $\mathbf{P}=(p_{i,j})$ in the following way:
\begin{eqnarray}
\label{set-pda}
p_{i,j}=\left\{\begin{array}{ll}
a & \  \mbox{if}\ \ (i,j,a)^T\in \mathcal{C},\\
* & \  \mbox{otherwise}.
\end{array}
\right.
\end{eqnarray}
Clearly there are $n$ integer entries, and the integer set is $[0,S)$. For any two distinct entries $p_{i_1,j_1}$ and $p_{i_2,j_2}$, assume that $p_{i_1,j_1}=p_{i_2,j_2}=s\in [0,S)$. Then
$i_1\ne i_2$ and $j_1\ne j_2$ hold from P1. So C1-a) holds. And $p_{i_1,j_2}=p_{i_2,j_1}=*$. Otherwise, without loss of generality, suppose that $p_{i_1,j_2}\in [0,S)$. Then by \eqref{set-pda}, there exists a subset $\triangle =\{(i_1,j_1,s)^T$, $(i_1,j_2,p_{i_1,j_2})^T$, $(i_2,j_2,s)^T\}\subseteq \mathcal{C}$, a contradiction to P2. So C1-b) holds.
\end{proof}

From Theorem \ref{the-pda-pda}, we can study PDA by means of discussing its incidence set $\mathcal{C}$. Given a set $\mathcal{C}$, conjugates of $\mathcal{C}$ are defined by rearranging the coordinates of $\mathcal{C}$. Let $\mathcal{C}_{(l_0,l_1,l_2)}$ be the conjugate set obtained by rearranging the coordinates of $\mathcal{C}$ in the order $(l_0,l_1,l_2)\in \mathcal{L}$, where $\mathcal{L}$ is the set formed by all the permutations of $[0,3)$.   For instance, $\mathcal{C}$ can be written as $\mathcal{C}_{(0,1,2)}$. $\mathcal{C}_{(2,1,0)}$ is obtained by changing the first coordinate and the third coordinate of $\mathcal{C}$. It is very easy to verify that the following statement holds.
\begin{lemma}
\label{le-C-conjugate}
$\mathcal{C}$ satisfies P1 and P2 if and only if its conjugates satisfy P1 and P2.
\end{lemma}

%Suppose that $\mathcal{C}$ is the incidence set of a PDA $\mathbf{P}$. For any conjugate set $\mathcal{C}_{(l_0,l_1,l_2)}$, we can define an array, say $\mathbf{P}_{(l_0,l_1,l_2)}=(p'_{i,j})$ by \eqref{set-pda}. For convenience, we also call $\mathbf{P}_{(l_0,l_1,l_2)}$ is a conjugate of $\mathbf{P}$.
%\begin{lemma}
%\label{le-PDA-conjugate}
%$\mathbf{P}$ is a PDA if and only if then its conjugates are also PDAs.
%\end{lemma}
%\begin{proof}
%Denote the incidence set of $\mathbf{P}$ by $\mathcal{C}$. From Theorem \ref{the-pda-pda}, $\mathcal{C}$ satisfies P1 and P2. Then the conjugate $\mathcal{C}_{{\bm l}}$ of $\mathcal{C}$ also satisfies P1 and P2 from Lemma \ref{le-C-conjugate}, where ${\bm l}\in \mathcal{L}$. From Theorem \ref{the-pda-pda}, $\mathbf{P}_{{\bm l}}$ generated by $\mathcal{C}_{{\bm l}}$ in \eqref{set-pda} is also a PDA. Similarly the converse also holds.
%\end{proof}
%
%From Lemma \ref{le-PDA-conjugate}, the following result can be obtained.
\begin{theorem}
\label{permutations of PDA}
Let $\mathbf{P}$ be a $(K,F,Z,S)$ PDA for some positive integers $K$, $F$, $Z$ and $S$ with $0< Z<F$. Then
\begin{itemize}
\item[1)] there exists an $(K,S,S-(F-Z),F)$ PDA;
\item[2)] if $\mathbf{P}$ is $g$-regular, then there exist an $(S,F,F-g,K)$ PDA and an $(S,K,K-g,F)$ PDA;
\item[3)] if each row has $h$ integer entries in $\mathbf{P}$, then there exist an $(F,S,S-h,K)$ PDA and an $(F,K,K-h,S)$ PDA.
\end{itemize}
\end{theorem}
\begin{proof}
Let $\mathcal{C}$ be the incidence set of $\mathbf{P}$. Now let us consider the conjugates of $\mathcal{C}$.
\begin{itemize}
\item When $(l_0,l_1,l_2)=(2,1,0)$, $\mathcal{C}_{(2,1,0)}\subseteq [0,S)\times [0,K)\times [0,F)$ satisfies P1 and P2 from Lemma \ref{le-C-conjugate}. Then $\mathbf{P}_{(2,1,0)}$ generated by $\mathcal{C}_{(2,1,0)}$ in \eqref{set-pda} is a $(K,S,F)$ PDA from Theorem \ref{the-pda-pda}. It is easy to check that each integer in $[0,K)$ occurs $F-Z$ times in the second coordinate of $\mathcal{C}$. This implies that each column of $\mathbf{P}_{(2,1,0)}$ has $S-(F-Z)$ stars. So $\mathbf{P}_{(2,1,0)}$ is a $(K,S,S-(F-Z),F)$ PDA.
\item When $(l_0,l_1,l_2)=(0,2,1)$, $\mathcal{C}_{(0,2,1)}\subseteq [0,F)\times [0,S)\times [0,K)$ satisfies P1 and P2 from Lemma \ref{le-C-conjugate}. Then $\mathbf{P}_{(0,2,1)}$ generated by $\mathcal{C}_{(0,2,1)}$ in \eqref{set-pda} is an $(S,F,K)$ PDA. If $\mathbf{P}$ is $g$-regular, i.e., each integer, say $s\in [0,S)$ occurs $g$ times in $\mathbf{P}$, then $s$ occurs $g$ times in the third coordinate of $\mathcal{C}$. This implies that each column of $\mathbf{P}_{(0,2,1)}$ has $F-g$ stars. So $\mathbf{P}_{(0,2,1)}$ is an $(S,F,F-g,K)$ PDA. Similarly we can show that $\mathbf{P}_{(1,2,0)}$ is an $(S,K,K-g,F)$ PDA.
\item When $(l_0,l_1,l_2)=(2,0,1)$, $\mathcal{C}_{(2,0,1)}\subseteq [0,S)\times [0,F)\times [0,K)$ satisfies P1 and P2 from Lemma \ref{le-C-conjugate}. Then
  $\mathbf{P}_{(2,0,1)}$ generated by $\mathcal{C}_{(2,0,1)}$ in \eqref{set-pda} is an $(F,S,K)$ PDA. If each row has $h$ integer entries in $\mathbf{P}$, i.e., each integer, say $f\in [0,F)$ occurs $h$ times in the first coordinate of $\mathcal{C}$, then each column of $\mathbf{P}_{(2,0,1)}$ has $h$ integers. That is, each column has $S-h$ stars. So $\mathbf{P}_{(2,0,1)}$ is an $(F,S,S-h,K)$ PDA. Similarly we can show that $\mathbf{P}_{(1,0,2)}$ is an $(F,K,K-h,S)$ PDA.
\end{itemize}
\end{proof}

From Theorems \ref{th-AN-Y} and \ref{permutations of PDA}, the following result can be obtained.
\begin{theorem}
\label{th-permutation of Ali}
For any positive integers $k$ and $t$ with $0<t<k-1$, we have the following PDAs.
\begin{itemize}
\item[(a)] $(t+1)$-$(k,{k\choose t},{k-1\choose t-1},{k\choose t+1})$ PDA with $t$ stars in each row;
\item[(b)] $(k,{k\choose t+1},{k-1\choose t+1},{k\choose t})$ PDA;
\item[(c)] $({k\choose t+1},{k\choose t},{k\choose t}-(t+1),k)$ PDA;
\item[(d)] $({k\choose t+1},k,k-(t+1),{k\choose t})$ PDA;
\item[(e)] $({k\choose t},{k\choose t+1},{k\choose t+1}-(k-t),k)$ PDA;
\item[(f)] $({k\choose t},k,t,{k\choose t+1})$ PDA.
%\item[(b)] $(k-t)$-$(k,{k\choose t+1},{k-1\choose t+1},{k\choose t})$ PDA with $k-t-1$ stars in each row;
%\item[(c)] $({k\choose t}-{k-1\choose t-1})$-$({k\choose t+1},{k\choose t},{k\choose t}-(t+1),k)$ PDA with ${k\choose t+1}-(k-t)$ stars in each row;
%\item[(d)] $(k-t)$-$({k\choose t+1},k,k-(t+1),{k\choose t})$ PDA with ${k-1\choose t+1}$ stars in each row;
%\item[(e)] $({k\choose t}-{k-1\choose t-1})$-$({k\choose t},{k\choose t+1},{k\choose t+1}-(k-t),k)$ PDA with ${k\choose t}-(t+1)$ stars in each row;
%\item[(f)] $(t+1)$-$({k\choose t},k,t,{k\choose t+1})$ PDA with ${k-1\choose t-1}$ stars in each row.
\end{itemize}
\end{theorem}

For $0<t'<k-1$, let $t=k-t'-1$. Clearly $0<t<k-1$. Applying Theorem \ref{th-permutation of Ali}-(a), \ref{th-permutation of Ali}-(c) and \ref{th-permutation of Ali}-(f) to $t=k-t'-1$, we can obtain $(k,{k\choose t+1},{k-1\choose t+1},{k\choose t})$ PDA, $({k\choose t},{k\choose t+1},{k\choose t+1}-(k-t),k)$ PDA and $({k\choose t+1},k,k-(t+1),{k\choose t})$ PDA, i.e., Theorem \ref{th-permutation of Ali}-(b), \ref{th-permutation of Ali}-(d) and \ref{th-permutation of Ali}-(e). So we only need to consider the PDAs in Theorem \ref{th-permutation of Ali}-(a), \ref{th-permutation of Ali}-(c) and \ref{th-permutation of Ali}-(f).

The subcase (a), i.e., MN PDA, in Theorem \ref{th-permutation of Ali} has been full discussed. Now let us consider the subcases (c) and (f). In the following, we will show that the PDAs in Theorem \ref{th-permutation of Ali}-(c) and \ref{th-permutation of Ali}-(f) are Pareto-optimal by the two lower bounds on the value of $S$ in Section \ref{Lower bounds}.
%Then two new coded caching schemes with good performance are obtained.

%In the practical setting, it is desirable to design caching schemes as
%well as small $F$ for the fixed $K$, $\frac{M}{N}$ and $R$. And this requirement is also proposed in \cite{SZG} as an open problem. In the following, we will show that the PDAs in Theorem \ref{th-permutation of Ali}-(a), \ref{th-permutation of Ali}-(c) and \ref{th-permutation of Ali}-(f) are optimal and have the minimum $F$s by the following lower bound on $S$.
%From Theorem \ref{th-AN-optimal PDA}, the PDA in Theorem \ref{th-permutation of Ali}-(a) has the minimum delivery rate and number of rows. So let us consider the other two PDAs respectively.

%In the practical setting, it is desirable to design caching schemes as
%well as small $F$ for the fixed $K$, $\frac{M}{N}$ and $R$. And this requirement is also proposed in \cite{SZG} as an open problem. In the following, we will show that the PDAs in Theorem \ref{th-permutation of Ali}-(a), \ref{th-permutation of Ali}-(c) and \ref{th-permutation of Ali}-(f) are optimal and have the minimum $F$s by the following lower bound on $S$.

\section{Two Pareto-optimal PDAs}
\label{se-other two good PDAs}
We first consider the PDAs in Theorem \ref{th-permutation of Ali}-(c) and \ref{th-permutation of Ali}-(f) using the lower bound on the value of $S$ respectively. Then we show that there exactly exists a tradeoff between rate $R$ and $F$ of the code caching schemes generated by the PDAs for some parameters.

\subsection{The first Pareto-optimal PDA}
\label{the first optimal PDA}

From Theorem \ref{th-digui-bound} and Theorem \ref{th-permutation of Ali}-(c), the following result can be obtained.
%\begin{lemma}
%\label{le-PDA_1}
%For any integers $k$ and $t$ with $0<t<k-1$, the $({k\choose t+1},{k\choose t},{k\choose t}-(t+1),k)$ PDA $\mathbf{P}_1$ in Theorem \ref{th-permutation of Ali}-(c) is optimal.
%\end{lemma}

\begin{theorem}
\label{th-r1-ali}
The $({k\choose t+1},{k\choose t},{k\choose t}-(t+1),k)$ PDA $\mathbf{P}_1$ in Theorem \ref{th-permutation of Ali}-(c) is Pareto-optimal.
%has the minimum number of rows when  $K={k\choose t+1}$, $\frac{M}{N}=1-\frac{t+1}{{k\choose t}}$ and $R\leq\frac{k}{{k\choose t}}$.
\end{theorem}
\begin{proof}
From Theorem \ref{thm1}, a coded caching scheme with $K={k\choose t+1}$, $M/N=1-\frac{t+1}{{k\choose t}}$, $F={k\choose t}$ and $R=\frac{k}{{k\choose t}}$ can be obtained by $\mathbf{P}_1$. First for any positive integer $F'< {k\choose t}$, assume that there exists a $(K,F',Z,S)$ PDA with $Z=F'(1-\frac{t+1}{{k\choose t}})$. We will show that the corresponding rate $R=S/F'> k/{k\choose t}$. According to \eqref{eq-digui-bound}
\begin{eqnarray*}
S\geq  \left\lceil\frac{(F'-Z)K}{F'}\right\rceil +F'-Z-1=\left\lceil\frac{\frac{F'(t+1)}{{k\choose t}}{k \choose t+1}}{F'}\right\rceil+\frac{F'(t+1)}{{k\choose t}}-1=k-t-1+\frac{F'(t+1)}{{k\choose t}}.
\end{eqnarray*}
So we have
$$R=\frac{S}{F'}\geq \frac{k-t-1+\frac{F'(t+1)}{{k\choose t}}}{F'}=\frac{k-t-1}{F'}+\frac{t+1}{{k\choose t}}>\frac{k}{{k\choose t}}.$$

Conversely we claim that for any positive real number $R<\frac{k}{{k\choose t}}$, the corresponding $F'>{k\choose t}$ always holds.  Assume that there exists a $({k\choose t+1},F',F'(1-\frac{t+1}{{k\choose t}}),K')$ PDA $\mathbf{P}'$ with $F'\leq {k\choose t}$ and $\frac{K'}{F'}<\frac{k}{{k\choose t}}$. Let $x=F'\frac{t+1}{{k\choose t}}$, i.e., $F'=x\frac{{k\choose t}}{t+1}$. It is easy to check that
the number of integer entries in $\mathbf{P}'$ is
$$n=\left(F'-F'\left(1-\frac{t+1}{{k\choose t}}\right)\right){k\choose t+1}=F'(k-t)={k \choose t+1}x.$$
Let $\mathcal{C}'$ be the incidence set of $\mathbf{P}'$. Then $\mathcal{C}'_{(0,2,1)}\subseteq [0,F')$ $\times$ $[0,K')$ $\times$ $[0,{k\choose t+1})$ satisfies P1 and P2 from Lemma \ref{le-C-conjugate}. So $\mathbf{P}'_{(0,2,1)}$ generated by $\mathcal{C}'_{(0,2,1)}$ in \eqref{set-pda} is a $(K',F',{k\choose t+1})$ PDA from Theorem \ref{the-pda-pda}. From \eqref{eq-Lowebound-K<F} we have
\begin{eqnarray}
\label{eq-P_1-1}
{k \choose t+1}\geq S'\geq\left\lceil\frac{nF'}{K'F'+F'-n}\right\rceil=
\left\lceil\frac{{k \choose t+1}xF'}{K'F'+F'-F'(k-t)}\right\rceil\geq
{k \choose t+1}\frac{x}{K'+1-(k-t)}.
\end{eqnarray}
So we have $\frac{x}{K'+1-(k-t)}\leq 1$. Clearly we have $K'+1-(k-t)>0$ since $K'F'+F'-n>0$. Then the following inequality holds
\begin{eqnarray}
\label{eq-P_1-2}
K'-k+(t+1)\geq x\geq 1.
\end{eqnarray}
From hypothesis $$\frac{k}{{k\choose t}}> \frac{K'}{F'}=\frac{K'}{\frac{x{k\choose t}}{t+1}}$$ we have
\begin{eqnarray*}
\frac{k}{t+1}> \frac{K'}{x}\geq \frac{K'}{K'-k+(t+1)}.
\end{eqnarray*}
Then $$\frac{t+1}{k}<  \frac{K'-k+(t+1)}{K'}.$$
That is $$\frac{k-t-1}{K'}< \frac{k-t-1}{k}.$$
This implies $K'> k$ since $t<k-1$. So we have $R=\frac{K'}{F'}>\frac{k}{{k\choose t}}$, a contradiction to our hypothesis $\frac{K'}{F'}< \frac{k}{{k\choose t}}$. Then we have $F> {k\choose t}$ if $R<\frac{k}{{k\choose t}}$.
\end{proof}
From Theorem \ref{thm1}, the parameters of the scheme generated by $\mathbf{P}_1$ can be obtain as follows.
\begin{eqnarray}\label{eq-p1-parameters}
\frac{M_1}{N_1}=1-\frac{t+1}{{k\choose t}},\ \ K_1={k\choose t+1},\ \ F_1={k\choose t},\ \ R_1=\frac{k}{{k\choose t}}.
\end{eqnarray}
From Theorem \ref{th-r1-ali}, this scheme has not only the low rate for the fixed parameters $K_1$, $\frac{M_1}{N_1}$ and $F_1$, but also the low packet number for the fixed parameters $K_1$, $\frac{M_1}{N_1}$ and $R_1$. In the following, by comparing the performance with MN PDA, we will show that there exits a tradeoff between $F$ and $R$ for the fixed parameters $K_1$ and $\frac{M_1}{N_1}$. When $K={k\choose t+1}$, $\frac{M}{N}=1-\frac{t+1}{{k\choose t}}$, from Theorem \ref{th-AN-Y} we have an MN PDA $(K$, $F_{MN}$, $Z_{MN}$, $S_{MN})$ PDA, where
$$F_{MN}={{k\choose t+1}\choose {k\choose t+1}(1-\frac{t+1}{{k\choose t}})}={{k\choose t+1}\choose k-t} \ \ \ \ \hbox{and}\ \ \ \ \ S_{MN}={{k\choose t+1}\choose {k\choose t+1}(1-\frac{t+1}{{k\choose t}})+1}={{k\choose t+1}\choose k-t-1}.$$
Then $$R_{MN}=\frac{S_{MN}}{F_{MN}}=\frac{k-t}{\binom{k}{t+1}-k+t+1}.$$
From \eqref{eq-p1-parameters}, we have
\begin{eqnarray}
\label{eq-r1-ali}
\frac{F_1}{F_{MN}} =\frac{{k\choose t}}{{{k\choose t+1}\choose k-t}}\leq\left(\frac{k}{{k\choose t+1}}\right)^{k-t}\ \ \ \ \hbox{and}\ \ \ \ \frac{R_1}{R_{MN}} =\frac{k}{t+1}-\frac{k}{{k\choose t}}+\frac{k}{(k-t){k\choose t}}.
\end{eqnarray}
The first item in \eqref{eq-r1-ali} is derived by the following fact.
\begin{eqnarray*}
\frac{{k\choose t}}{{{k\choose t+1}\choose k-t}}&=&\frac{\frac{k(k-1)\ldots(t+1)}{(k-t)!}}{\frac{{k\choose t+1}({k\choose t+1}-1)\ldots({k\choose t+1}-k+t+1)}{(k-t)!}}=\frac{k(k-1)\ldots(t+1)}{{k\choose t+1}({k\choose t+1}-1)\ldots({k\choose t+1}-k+t+1)}\\
&=&\frac{k}{{k\choose t+1}}\frac{k-1}{{k\choose t+1}-1}\ldots\frac{t+1}{{k\choose t+1}-k+t+1}\\
&\leq&\left(\frac{k}{{k\choose t+1}}\right)^{k-t}
\end{eqnarray*}
The last inequality of the above formula holds due to $\frac{k}{{k\choose t+1}}\geq\frac{k-x}{{k\choose t+1}-x}$ for any positive integer $x\in[1,k-t)$.
\begin{remark}
\label{re4}
MN PDA with the parameters $K_1$ and $\frac{M_1}{N_1}$ in \eqref{eq-p1-parameters} could achieve the minimum rate, but its $F$ is at least $F={{k\choose t+1}\choose k-t}$. By \eqref{eq-r1-ali}, if $F$ reduces by more than $({k\choose t+1}/{k})^{k-t}$ times, the rate must increase at least $\frac{k}{t+1}-\frac{k}{{k\choose t}}+\frac{k}{(k-t){k\choose t}}$ times. In other words, if $R$ increases by a factor of $\frac{k}{t+1}$ times, then $F$ could decrease by more than $({k\choose t+1}/{k})^{k-t}$ times.
\end{remark}
\begin{example}
Let $t=k-3$. By \eqref{eq-r1-ali} we have
\begin{eqnarray*}
\frac{F}{F_{MN}}=\frac{8k-16}{\left((k-1) k-4\right) \left((k-1) k-2\right)}\ \ \ \hbox{and} \ \ \
\frac{R}{R_{MN}}=\frac{k}{k-2}-\frac{4}{(k-1)(k-2)}.
\end{eqnarray*}
According to above formula, the following table can be obtained.
\begin{eqnarray*}
\begin{array}{|c|c|c|c|c|c|c|c|c|c|c|c|c|c|c|c|c|c|c|c|}
\hline
k                 & 5 &  6 &  7 &  8 &  9 & 10 & 11 & 12 &  13 & 14  &  15\\ \hline
K                 & 10 & 20 & 35 & 56 & 84 & 120 & 165 & 220 & 286 & 364 & 455\\ \hline
\frac{F}{F_{MN}} & 0.194 & 0.088 & 0.047 & 0.028 & 0.018 & 0.013 & 0.009 & 0.007 & 0.005 & 0.004  & 0.003\\ \hline
\frac{R}{R_{MN}} & 1.333 & 1.3   & 1.267 & 1.238 & 1.214 &1.194  & 1.178 & 1.164 & 1.152 & 1.141 & 1.132\\
\hline
\end{array}
\end{eqnarray*}
\end{example}
\subsection{The second optimal PDA}
\label{the second optimal PDA}

\begin{theorem}
\label{th-r2-ali}
The $({k\choose t},k,t,{k\choose t+1})$ PDA $\mathbf{P}_2$ in Theorem \ref{th-permutation of Ali}-(f) is Pareto-optimal.
%for fixed parameters $K={k\choose t}$, $\frac{M}{N}=\frac{t}{k}$ and $R\leq\frac{{k\choose t+1}}{k}$.
\end{theorem}
\begin{proof} From Theorem \ref{thm1}, a coded caching scheme with $K={k\choose t}$, $M/N=t/k$, $F=k$ and $R=\frac{{k\choose t+1}}{k}$ can be obtained by $\mathbf{P}_2$. Clearly $F\leq K$. So we only need to consider the relationship between $R$ and $F$ in \eqref{eq-tradeoff}. It is easy to check that \eqref{eq-tradeoff} is a monotonic decreasing function of $F$. So we only need to check whether the equality in \eqref{eq-tradeoff} holds for the parameters $K$, $F$, $M/N$ and $R$. Then the result directly follows from \eqref{eq-tradeoff}, i.e.,
\begin{eqnarray*}
R=\frac{{k\choose t+1}}{k}\geq\frac{K(1-M/N)}{FM/N+1}=\frac{{k\choose t}\frac{k-t}{k}}{k\frac{t}{k}+1}=\frac{{k\choose t+1}}{k}.
\end{eqnarray*}
\end{proof}
From Theorem \ref{thm1}, the parameters of the scheme generated by $\mathbf{P}_2$ can be obtain as follows.
\begin{eqnarray}\label{eq-p2-parameters}
\frac{M_2}{N_2}=\frac{t}{{k}},\ \ K_2={k\choose t},\ \ F_2=k,\ \ R_2=\frac{{k\choose t+1}}{k}.
\end{eqnarray}
From Theorem \ref{th-r2-ali}, this scheme has not only the low rate for the fixed parameters $K_2$, $\frac{M_2}{N_2}$ and $F_2$, but also the low packet number for the fixed parameters $K_2$, $\frac{M_2}{N_2}$ and $R_2$. Similarly by comparing the performance with MN PDA in the above subsection, we will also show that there exactly exists a tradeoff between $F$ and $R$ for the fixed parameters $K_2$ and $\frac{M_2}{N_2}$ in \eqref{eq-p2-parameters}. When $K={k\choose t}$ and $\frac{M}{N}=\frac{t}{k}$, from Theorem \ref{th-AN-Y} we have an MN PDA $(K$, $F_{MN}$, $Z_{MN}$, $S_{MN})$ PDA $\mathbf{P}$, where
$$F_{MN}={{k\choose t}\choose {k\choose t}\frac{t}{k}}={{k\choose t}\choose {k-1\choose t-1}}\ \ \ \ \hbox{and}\ \ \ \ S_{MN}={{k\choose t}\choose {k\choose t}\frac{t}{k}+1}.$$
Then
$$ R_{MN}=\frac{S_{MN}}{F_{MN}}=\frac{k-t}{t{k\choose t}+k}{k\choose t}.$$
From \eqref{eq-p2-parameters}, we have
\begin{eqnarray}
\label{eq-r2-ali}
\frac{F_2}{F_{MN}} =\frac{k}{{{k\choose t}\choose {k-1\choose t-1}}}\leq k\left(\frac{t}{k}\right)^{{k-1\choose t-1}}\ \ \ \ \hbox{and}\ \ \ \ \frac{R_2}{R_{MN}} =\frac{{k \choose t+1}/{k}}{\frac{k-t}{t{k\choose t}+k}{k\choose t}}=\left({k-1\choose t-1}+1\right)\frac{1}{t+1}
\end{eqnarray}
where the first inequality holds due to $\frac{{k\choose t}-x}{{k-1\choose t-1}-x}\geq\frac{{k\choose t}}{{k-1\choose t-1}}=   \frac{k}{t}$ for any integer $x\in [1,{k-1\choose t-1})$.
\begin{remark}
MN PDA with the parameters $K_2$ and $\frac{M_2}{N_2}$ in \eqref{eq-p2-parameters} could achieve the minimum rate, but its $F$ is at least ${{k\choose t}\choose {k-1\choose t-1}}$. From \eqref{eq-r2-ali}, if $F$ reduces by more than $(\frac{k}{t})^{{k-1\choose t-1}}$ times, the rate must increase at least ${k-1\choose t-1}$ times. In other words, if $R$ increases about ${k-1\choose t-1}$ times, then $F$ could decrease by more than $(\frac{k}{t})^{{k-1\choose t-1}}$ times.
\end{remark}
\begin{example}
When $t=2$, by \eqref{eq-r2-ali}, we have
$$
\frac{F_2}{F_{MN}}=\frac{k}{{{k\choose t}\choose {k-1\choose t-1}}}=\frac{k}{{{k\choose 2}\choose k-1}}\leq\frac{k}{(\frac{k}{2})^{k-1}}\ \ \ \ \hbox{and}\ \ \ \ \frac{R_2}{R_{MN}} =\frac{k}{3},
$$
where the first item holds due to $(\frac{m}{l})^{l}\le{m\choose l}$ for any positive integers $m$ and $l$ with $l\leq m$. Then the following table can be obtained.
\begin{eqnarray*}
\begin{array}{|c|c|c|c|c|c|c|c|c| }
\hline
k& 4     &5     &    6 & 7    &8    &  9 &  10   \\ \hline
\frac{F_2}{F_{MN}}
 &0.2000&2.3810\cdot 10^{-2}&1.9980\cdot 10^{-3}&1.2900\cdot 10^{-4}&6.7565\cdot 10^{-6}&2.9742\cdot 10^{-7}&1.1285\cdot 10^{-8}\\ \hline
\frac{R_2}{R_{MN}}
  &1.3  &1.7  &2.0 &2.3    &2.7   &3.0      &3.3 \\ \hline
\end{array}
\end{eqnarray*}
\end{example}

\section{The performance of previously known PDAs}
\label{sec-comparison}
In fact, comparing the two Pareto-optimal PDAs in the above section, we can estimate the performance of some previously known results. In this section, we will have the following statements.
\begin{itemize}
\item {\bf Comparison 1:} For the same parameters $K$, $M/N$ and $R$, the packet number of the scheme generated by $\mathbf{P}_2$ is far less than that of PDA in Lemma \ref{lem-shanmugam}. This implies that the PDA in Lemma \ref{lem-shanmugam} is not Pareto-optimal for some $K$ and $M/N$;
\item {\bf Comparison 2:} The packet number of the scheme generated by $\mathbf{P}_1$ significantly decrease by increasing a little delivery rate $R$ comparing with that of the PDA in Lemma \ref{lem-yan};
\item {\bf Comparison 3:} For some parameters $K$ and $M/N$, the performance of $\mathbf{P}_1$ is better than that of the PDA in Lemma \ref{lem-shang}, i.e., both the rate and the packet number of the scheme generated by $\mathbf{P}_1$ are smaller than that of the PDA in Lemma \ref{lem-shang}. This implies that the PDA in Lemma \ref{lem-shang} is not Pareto-optimal for some $K$ and $M/N$;
\end{itemize}
Now let us introduce above three comparisons one by one.
\subsection{Comparison 1}
For any positive integers $K$ and $K'$ with $K'|K$, let $m=K/K'$. From Lemma \ref{lem-shanmugam}, we have a $(K,{K' \choose KM/N},{K'-1 \choose KM/N-1}$, $m{K'\choose KM/N+1})$ PDA. From Theorem \ref{thm1},
\begin{eqnarray*}
F'={K'\choose K'M/N} \ \ \ \ \ \hbox{and}\ \ \ \ R'=\frac{m{K'\choose KM/N+1}}{{K' \choose KM/N}}=\frac{K'(1-\frac{M}{N})}{1+K'\frac{M}{N}}\frac{K}{K'}.
\end{eqnarray*}

%In fact the case $K'\not|K$ is also discussed in \cite{SJTLD}. Since the discussion is same as the first case, we only consider the case $K'|K$ for easy comparison.
For any positive integers $k$ and $t$, now let us consider the performance of the PDA in Lemma \ref{lem-shanmugam} and Theorem \ref{th-r2-ali}.  Assume
$$K={k\choose t}\ \ \ \ \ \hbox{and}\ \ \ \ \frac{M}{N}=\frac{t}{k}.$$
From Theorems \ref{thm1} and \ref{th-r2-ali}, we have
\begin{eqnarray*}
F_2=k \ \ \ \ \ \hbox{and}\ \ \ \ R_2=\frac{{k\choose t+1}}{k}=\frac{k-t}{k(1+t)}{k\choose t}.
\end{eqnarray*}
Assume that $R'=R_2$, i.e.,
\begin{eqnarray*}
\frac{k-t}{k(1+t)}{k\choose t}=\frac{K'(1-\frac{t}{k})}{1+K'\frac{t}{k}}\frac{K}{K'}
=\frac{k-t}{k+K't}{k\choose t}.
\end{eqnarray*}
Then we have $K'=k$. So we have
$$\frac{F_2}{F'}=\frac{k}{{k\choose t}}.$$
\begin{remark}
\label{re-group-p2}
From the above discussion, for same parameters $K={k\choose t}$, $\frac{M}{N}=\frac{t}{k}$ and $R={k\choose t+1}/k$, the row number $F'$ of the PDA in Lemma \ref{lem-shanmugam} is ${k\choose t}/k$ times larger than that of $\mathbf{P}_2$ in Theorem \ref{th-permutation of Ali}-(c), where $k$, $t$ are any positive integers with $t<k$.  So the PDA in Lemma \ref{lem-shanmugam} is not Pareto-optimal when any $1<t<k-1$.
\end{remark}

%Given the parameters $K$ and $M/N$, several comparisons with MN scheme were proposed in \cite{SZG} and \cite{YCTC}. So in this section, the following comparisons between scheme generated by $\mathbf{P}_1$, $\mathbf{P}_2$ and previously known results are considered. Since all the previously determined scheme can be represented by PDAs, we only need to consider the corresponding PDAs.

\subsection{Comparasion 2}
\label{sub-comparison 1}
Let us consider the $(q(m+1),(q-1)q^m,(q-1)^2q^{m-1},q^m)$ PDA $\mathbf{P}_3$ in Lemma \ref{lem-yan}. From Theorem \ref{thm1}, the parameters of the scheme generated by $\mathbf{P}_3$ can be obtained as follows.
$$\frac{M_3}{N_3}=1-\frac{1}{q},\ \ K_3=q(m+1),\ \ F_3=q^m(q-1),\ \ R_3=\frac{1}{q-1}$$
%From Theorem \ref{th-r1-ali}, in $\mathbf{P}_1$ we have
%\begin{eqnarray}\label{eq-p1-parameters}
%\frac{M_1}{N_1}=1-\frac{t+1}{{k\choose t}},\ \ K_1={k\choose t+1},\ \ F_1={k\choose t},\ \ R_1=\frac{k}{{k\choose t}}.
%\end{eqnarray}
From \eqref{eq-p1-parameters}, assume that $\frac{M_3}{N_3}= \frac{M_1}{N_1}$ and $K_3= K_1$, i.e.,
$$1-\frac{1}{q}=1-\frac{t+1}{{k\choose t}}\ \ \ \hbox{and}\ \ \ {k\choose t+1}=q(m+1),$$
for some positive integers $m$, $t$ and $k$. Then we have $$m+1=k-t$$
and
\begin{eqnarray}
\label{eq-Comparison-partition and P1}
%\begin{array}{ccc}
\frac{F_1}{F_3}=\frac{{k\choose t}}{q^m(q-1)}=\frac{{k\choose t}}{\Big(\frac{{k\choose t}}{t+1}\Big)^{k-t-1}\Big(\frac{{k\choose t}}{t+1}-1\Big)},\ \ \ \ \ \
%\le \frac{t+1}{k^{k-t-2}(k-1)}\\
%\end{array}\ \ \ \
%\begin{array}{ccc}
 \frac{R_1}{R_3}=\frac{(q-1)k}{{k\choose t}}=\frac{k}{t+1}-\frac{k}{{k\choose t}}.
%\end{array}
\end{eqnarray}
Now let us consider the values of ${F_1}/{F_3}$ and ${R_1}/{R_3}$ according to $0<t<k-1$.
\begin{itemize}
%\item When $t=k-1$, it is easy to check that $\mathbf{P}_1$ and  $\mathbf{P}_3$ are trivial.
\item When $t=k-2$, \eqref{eq-Comparison-partition and P1} can be written in the following way.
$$\frac{F_1}{F_3}=\frac{2 (k-1)}{k-2}\ \ \ \ \ \ \ \ \frac{R_1}{R_3}=\frac{k-2}{k-1}$$
\item When $t\le k-3$, \eqref{eq-Comparison-partition and P1} can be written in the following way.
\begin{eqnarray}
\label{eq-comp-yan-p1}
\frac{F_1}{F_3}\leq \frac{t+1}{k^{k-t-2}(k-1)}\ \ \ \ \ \ \ \ \frac{R_1}{R_3}=\frac{k}{t+1}-\frac{k}{{k\choose t}}
\end{eqnarray}
The first above item is derived by the following fact when $3\leq t\leq k-3$.
\begin{eqnarray*}
\frac{{k\choose t}}{t+1}&=&\frac{k(k-1)\ldots(t+1)}{(k-t)!(t+1)}=k\frac{k-1}{t}\frac{k-2}{t-1}\ldots\frac{t+3}{3\cdot 2}\frac{t+2}{4}\geq k
\end{eqnarray*}
\end{itemize}
\begin{remark}
\label{re-partition-p1}
By comparing $\mathbf{P}_3$ in Lemma \ref{lem-yan} above, the row number of $\mathbf{P}_1$ in Theorem \ref{th-permutation of Ali}-(c) decreases by more than $k^{k-t-2}$ times, but $R$ increases about $\frac{k}{t+1}$ times for some parameters $t$ and $k$.
\end{remark}
\begin{example}
\label{ex-yan-p1}
When $K_1=K_3$ and $\frac{M_3}{N_3}=\frac{M_1}{N_1}$, the following table can be obtained for some small positive integers $k$, $t$, $m$ and $q$.
\begin{eqnarray*}
\begin{array}{|c|c|c|c|c|c|}
\hline
k & t  &  m & q   &F_1/F_3       & R_1/R_3   \\ \hline
6 &  2 &  3 &  5  &  0.03        &  1.6\\ \hline
6 &  3 &  2 &  5  &  0.2         &  1.2\\ \hline
7 &  2 &  4 &  7  &  0.001457    &  2\\ \hline
7 &  4 &  2 &  7  &  0.119048    &  1.2\\ \hline
8 &  3 &  4 & 14  &  0.000112    &  1.85714 \\ \hline
8 &  4 &  3 & 14  &  0.001962    &  1.48571 \\ \hline
9 &  2 &  6 & 12  &  0.000001    &  2.75 \\ \hline
9 &  6 &  2 & 12  &  0.05303     &  1.17857 \\ \hline
10&  2 &  7 & 15  & 0.00000002   &  3.11111\\ \hline
\end{array}
\end{eqnarray*}
\end{example}

\subsection{Comparison 3}

Finally let us consider the $({m\choose l}q^l,q^m(q-1)^l,(q^m-q^{m-l})(q-1)^l,q^m)$PDA $\mathbf{P}_4$ in Lemma \ref{lem-shang} for any positive integers $q\geq 2$, $l$ and $m$ with $l\le m$. From Theorem \ref{thm1}, the parameters of the scheme generated by $\mathbf{P}_4$ can be obtain as follows.
$$\frac{M_4}{N_4}=1-\frac{1}{q^l},\ \ K_4={m\choose l}q^l,\ \ F_4=q^m(q-1)^l,\ \ R_4=\frac{1}{(q-1)^l}.$$

From \eqref{eq-p1-parameters}, assume that $\frac{M_1}{N_1}\leq \frac{M_4}{N_4}\ \hbox{and}\ K_1\geq K_4$, i.e.,
$$1-\frac{t+1}{{k\choose t}}\leq 1-\frac{1}{q^l}\ \hbox{and}\ {k\choose t+1}\geq q^l{m\choose l}$$ for some positive integers $m$, $t$, $l$ and $k$ with $t\le k$ and $l\leq m$. So we have
\begin{eqnarray}
\label{eq-Ge-P1-K,MN}
q^l(t+1)\geq {k\choose t}\ \ \ \ \ \ \hbox{and}\ \ \ \ {k\choose t}\geq q^l{m\choose l}\frac{t+1}{k-t}.
\end{eqnarray}
This implies
$$\left(\frac{{k\choose t}}{t+1}\frac{k-t}{{m\choose l}}\right)^{1/l}\geq q\geq \left(\frac{{k\choose t}}{t+1}\right)^{1/l}\ \ \ \ \ \ \hbox{and}\ \ \ \  \ \ \ \            {m\choose l}\leq k-t.$$
Then we have
\begin{eqnarray}\label{eq-Comparison-Ge and P1-1}
\frac{F_1}{F_4}=\frac{{k\choose t}}{q^m(q-1)^l}\leq\frac{q^l(t+1)}{q^m (q-1)^l}=\frac{t+1}{q^{m-l} (q-1)^l}\end{eqnarray}
\begin{eqnarray}\label{eq-Comparison-Ge and P1-2}
\frac{R_1}{R_4}=\frac{k (q-1)^l}{{k\choose t}}\leq \frac{k (q-1)^l}{q^l{m\choose l}\frac{t+1}{k-t}}=\Big(\frac{q-1}{q}\Big)^l \frac{k}{t+1}\frac{k-t}{{m\choose l}}.
\end{eqnarray}
With the aid of a computer, we can find out some parameters $K$, $t$, $m$ and $l$ listed in the following example satisfying \begin{eqnarray}\label{eq-Comparison-Ge and P1-good}
\frac{K_1}{K_4}\geq 1,  \ \ \ \frac{M_1}{N_1}\leq \frac{M_4}{N_4}, \ \ \ \frac{F_1}{F_4}< 1\ \ \ \hbox{and}\ \ \ \frac{R_1}{R_4}<1.
\end{eqnarray}
\begin{example}
Let $l=m-1$ in Lemma \ref{lem-shang}. The following table can be obtained by \eqref{eq-Comparison-Ge and P1-1} and \eqref{eq-Comparison-Ge and P1-2}.
\label{ex-ge-P1}
\begin{eqnarray*}
\begin{array}{|c|c|c|c|c|c|c|c| }
\hline
k    & t   &m   & q   & K_1/K_4&\frac{M_1}{N_1}/\frac{M_4}{N_4}    &F_1/F_4  & R_1/R_4   \\ \hline
 7   & 3   & 3  & 3   & 1.2963 & 0.9964 & 0.6481 & 0.8000\\ \hline
 25  & 22  & 3  & 10  & 7.6667 & 1      & 0.2556 & 0.8804\\ \hline
% 9   & 3   & 4  & 3   & 0.7778 & 0.9890 & 0.5185 & 0.8571\\ \hline
 9   & 4   & 4  & 3   & 1.1667 & 0.9973 & 0.7778 & 0.5714\\ \hline
 13  & 7   & 4  & 6   & 1.9861 & 1      & 0.2648 & 0.9470\\ \hline
 14  & 9   & 4  & 6   & 2.3171 & 0.9996 & 0.3090 & 0.8741\\ \hline
 17  & 12  & 4  & 8   & 3.0215 & 0.9999 & 0.2158 & 0.9423\\ \hline
 20  & 15  & 4  & 10  & 3.8760 & 1      & 0.1723 & 0.9404\\ \hline
% 11  & 4   & 5  & 3   & 0.8148 & 0.9972 & 0.6790 & 0.5333\\ \hline
 11  & 5   & 5  & 3   & 1.1407 & 0.9994 & 0.9506 & 0.3810\\ \hline
% 12  & 3   & 5  & 3   & 0.5432 & 0.9941 & 0.4527 & 0.8727\\ \hline
% 13  & 3   & 5  & 3   & 0.7062 & 0.9983 & 0.5885 & 0.7273\\ \hline
 13  & 5   & 5  & 4   & 1.0055 & 0.9992 & 0.4189 & 0.8182\\ \hline
 13  & 6   & 5  & 4   & 1.3406 & 0.9998 & 0.5586 & 0.6136\\ \hline
 13  & 7   & 5  & 4   & 1.3406 & 0.9992 & 0.5586 & 0.6136\\ \hline
 16  & 9   & 5  & 6   & 1.7654 & 0.9999 & 0.2942 & 0.8741\\ \hline
 18  & 10  & 5  & 8   & 2.1366 & 1      & 0.1908 & 0.9877\\ \hline
 19  & 12  & 5  & 8   & 2.4604 & 1      & 0.2197 & 0.9054\\ \hline
 23  & 17  & 5  & 9   & 3.0772 & 1      & 0.2137 & 0.9332\\ \hline
 25  & 19  & 5  & 10  & 3.5420 & 1      & 0.1968 & 0.9262\\ \hline
% 13  & 5   & 6  & 3  & 0.8827 & 0.9995 & 0.8827 & 0.3232\\ \hline
% 14  & 4   & 6  & 3  & 0.6866 & 0.9991 & 0.6866 & 0.4476\\ \hline
% 15  & 6   & 6  & 4  & 0.8146 & 0.9996 & 0.4073 & 0.7283\\ \hline
 15  & 7   & 6  & 4  & 1.0474 & 0.9997 & 0.5237 & 0.5664\\ \hline
% 16  & 3   & 6  & 3  & 0.3841 & 0.9970 & 0.3841 & 0.9143\\ \hline
% 16  & 5   & 6  & 4  & 0.7109 & 0.9996 & 0.3555 & 0.8901\\ \hline
% 17  & 3   & 6  & 3  & 0.4664 & 0.9982 & 0.4664 & 0.8000\\ \hline
 17  & 7   & 6  & 5  & 1.0372 & 0.9999 & 0.3112 & 0.8951\\ \hline
 17  & 8   & 6  & 5  & 1.2965 & 0.9999 & 0.3890 & 0.7161\\ \hline
 17  & 9   & 6  & 5  & 1.2965 & 0.9999 & 0.3890 & 0.7161\\ \hline
% 18  & 3   & 6  & 3  & 0.5597 & 0.9992 & 0.5597 & 0.7059\\ \hline
% 18  & 6   & 6  & 5  & 0.9901 & 0.9999 & 0.2970 & 0.9929\\ \hline
% 19  & 3   & 6  & 3  & 0.6646 & 1  & 0.6646 & 0.6275\\ \hline
 19  & 11  & 6  & 6  & 1.6200 & 1  & 0.3240 & 0.7856\\ \hline
 20  & 9  & 6  & 7  & 1.6656 & 1  & 0.2379 & 0.9259\\ \hline
 20  & 10  & 6  & 7  & 1.8321 & 1  & 0.2617 & 0.8418\\ \hline
 20  & 11  & 6  & 7  & 1.6656 & 1  & 0.2379 & 0.9259\\ \hline
 21  & 13  & 6  & 7  & 2.0179 & 1  & 0.2883 & 0.8025\\ \hline
 21  & 14  & 6  & 6  & 2.4923 & 1  & 0.4985 & 0.5644\\ \hline
 23  & 14  & 6  & 9  & 2.3065 & 1  & 0.1922 & 0.9223\\ \hline
 23  & 15  & 6  & 8  & 2.4939 & 1  & 0.2672 & 0.7884\\ \hline
 23  & 16  & 6  & 7  & 2.4311 & 1  & 0.3473 & 0.7295\\ \hline
 26  & 18  & 6  & 10  & 2.6038 & 1  & 0.1736 & 0.9827\\ \hline
 28  & 21  & 6  & 9  & 3.3420 & 1  & 0.2785 & 0.7749\\ \hline
\end{array}
\end{eqnarray*}
\end{example}

In fact, we can find out several classes of the parameters $k$, $t$, $m$ and $l$ with $0<t<k-1$ and $1\leq l<m$ such that \eqref{eq-Comparison-Ge and P1-good} holds. First we can obtain that
\begin{eqnarray}
\label{eq-compa-11}
q\geq \left(\frac{{k\choose t}}{t+1}\right)^{1/l}\geq \left(\frac{k}{k-t}\right)^{\frac{k-t-1}{l}}
\end{eqnarray}

since
\begin{eqnarray*}
\frac{{k\choose t}}{t+1}=\frac{{k\choose k-t}}{t+1}=\frac{k(k-1)\ldots(t+1)}{(k-t)!(t+1)}=\frac{k}{k-t}\frac{k-1}{k-t-1}\ldots\frac{t+2}{2}\frac{t+1}{t+1}\geq \left(\frac{k}{k-t}\right)^{k-t-1}
\end{eqnarray*}
The last item is derived by the fact that $\frac{k}{k-t}<\frac{k-x}{k-t-x}$ where $1\leq x<k-t$.

Moreover we have $k-t-1\geq l$ since
$$k-t\geq {m\choose l}\geq m\geq l+1.$$
This implies that $q> \frac{k}{k-t}$ from \eqref{eq-compa-11}. Submitting this inequality into \eqref{eq-Comparison-Ge and P1-1} and \eqref{eq-Comparison-Ge and P1-2}, the following results can be obtained.
\begin{eqnarray}\label{eq-Comparison-Ge and P1-m}
\frac{F_1}{F_4}< (t+1)\left(\frac{k-t}{t}\right)^l \left(\frac{k-t}{k}\right)^{m-l}\ \ \ \ \ \
\frac{R_1}{R_4}<  \left(\frac{k-t}{t}\right)^l \frac{k}{t+1}\frac{k-t}{{m\choose l}}.
\end{eqnarray}
Clearly from \eqref{eq-Ge-P1-K,MN} and \eqref{eq-Comparison-Ge and P1-m}, for some parameters $k$, $t$, $m$ and $l$, \eqref{eq-Comparison-Ge and P1-good} always holds. For instance, when $41\leq t=k-6$ and $1<l<m$, from \eqref{eq-Comparison-Ge and P1-m} we have
\begin{eqnarray*}
&&\frac{F_1}{F_4}< (k-5)\left(\frac{6}{k-6}\right)^l \left(\frac{6}{k}\right)^{m-l}=\frac{36(k-5)}{k(k-6)}\left(\frac{6}{k-6}\right)^{l-1} \left(\frac{6}{k}\right)^{m-l-1}\leq\frac{36(k-5)}{k(k-6)},\\
&&\frac{R_1}{R_4}< \left(\frac{6}{k-6}\right)^l \frac{k}{k-5}\frac{6}{{m\choose l}}<\frac{36k}{(k-5)(k-6)}\left(\frac{6}{k-6}\right)^{l-1}\leq \frac{36k}{(k-5)(k-6)}.
\end{eqnarray*}
Clearly when $0<\frac{36(k-5)}{k(k-6)}<\frac{36k}{(k-5)(k-6)}<1$, i.e., $k>\frac{1}{2} \left(47+\sqrt{2089}\right)=46.352$, $F_1/F_4<1$ and $R_1/R_4<1$ always hold.
%\begin{eqnarray}
%\label{eq-Comparison-Ge and P1}
%\begin{split}
%&\frac{F_1}{F_5}=\frac{{k\choose t}}{q^m(q-1)^s}\leq\frac{{k\choose t}}{\left(\frac{{k\choose t}}{t+1}\right)^{m/s} \left(\left(\frac{{k\choose t}}{t+1}\right)^{1/s}-1\right)^s}=\frac{t+1}{\left(\frac{{k\choose t}}{t+1}\right)^{m/sn-1} \left(\left(\frac{{k\choose t}}{t+1}\right)^{1/sn}-1\right)^s}\\
%&\frac{R_1}{R_5}=\frac{k (q-1)^s}{{k\choose t}}\leq \frac{k \left(\left(\frac{{k\choose t}}{t+1}\frac{k-t}{{m\choose s}}\right)^{1/s}-1\right)^s}{{k\choose t}}%\ \ \ \ =\Big(\frac{q-1}{q}\Big)^s \frac{k}{t+1}\frac{k-t}{{m\choose s}}.
%\end{split}\end{eqnarray}
%When $s=2$, \eqref{eq-Comparison-Ge and P1} can be written as follows.
%\begin{eqnarray}
%\label{eq-Comparison-Ge and P1-s=2}
%\begin{split}
%&\frac{F_1}{F_5}\leq \frac{t+1}{\left(\frac{{k\choose t}}{t+1}\right)^{m/2-1} \left(\frac{{k\choose t}}{t+1}-2\sqrt{\frac{{k\choose t}}{t+1}}+1\right)} \\
%&\frac{R_1}{R_5}\leq \frac{k \left(\frac{{k\choose t}}{t+1}\frac{k-t}{{m\choose s}}-2\sqrt{\frac{{k\choose t}}{t+1}\frac{k-t}{{m\choose s}}}+1-1\right)}{{k\choose t}}=
%\frac{k}{t+1}\frac{k-t}{{m\choose s}}-2k\sqrt{\frac{1}{t+1}\frac{k-t}{{m\choose s}}\frac{1}{{k\choose t}}}+\frac{k}{{k\choose t}}
%%\ \ \ \ =\Big(\frac{q-1}{q}\Big)^s \frac{k}{t+1}\frac{k-t}{{m\choose s}}.
%\end{split}\end{eqnarray}

\begin{remark}
\label{re-partition-p1}
By the above comparison, both the parameters $F$ and $R$ of $\mathbf{P}_1$ in Theorem \ref{th-permutation of Ali}-(c) are smaller than that of $\mathbf{P}_4$ in Lemma \ref{lem-shang} even if
$K\geq K_4$ and $M_1/N_1\leq M_4/N_4$ for some positive integers $k$, $t$, $m$, $l$ and $q$. So the PDA in Lemma \ref{lem-shang} is not Pareto-optimal for some parameters $K$ and $M/N$.
\end{remark}

\section{Conclusion}
\label{conclusion}
In this paper, we considered Pareto-optimal PDAs for some parameters $K$ and $M/N$. Firstly two lower bounds on the value of $R$ were derived. Consequently, MN PDA was pareto optimal, and a tradeoff between the lower bound on $R$ and $F$ was obtained for some parameters. Secondly, unlike the previously know strategies of constructing PDAs, we used a different strategy, i.e., we characterized PDAs by means of a set of $3$ dimensional vectors. From this characterization and the above lower bounds, two new classes of Pareto-optimal PDAs were obtained. Based on these two new PDAs, the schemes with low rate and packet number were obtained. Finally the performance of three previously known schemes were estimated by comparing with these two classes of PDAs.

It is interesting to prove that one class of the other previously known PDAs is Pareto-optimal for some parameters $K$ and $M/N$. And it would be meaningful to characterize the tradeoff between $R$ and $F$ for the other parameters $K$ and $M/N$ and construct the related Pareto-optimal PDAs.

\end{document}